\documentclass[10pt,a4paper]{article}
\usepackage[utf8]{inputenc}
\usepackage{amsmath, amsfonts, amssymb, graphicx, enumerate, framed, fullpage, calc, subfig, tikz, enumitem, mathtools, booktabs, hyphenat}
\usepackage[amsthm,thmmarks]{ntheorem}
\usepackage[numbers]{natbib}

\usepackage[small]{complexity}

\makeatletter
\g@addto@macro\@floatboxreset\centering
\makeatother

\newcount\bsubfloatcount
\newtoks\bsubfloattoks
\newdimen\bsubfloatht

\makeatletter
\newcommand{\bsubfloat}[2][]{%
  \sbox\z@{#2}%
  \ifdim\bsubfloatht<\ht\z@
    \bsubfloatht=\ht\z@
  \fi
  \advance\bsubfloatcount\@ne
  \@namedef{bsubfloat\romannumeral\bsubfloatcount}{%
    \subfloat[#1]{\vbox to\bsubfloatht{\hbox{#2}\vfill}}}%
}
\newcommand{\resetbsubfloat}{\bsubfloatcount\z@\bsubfloatht=\z@}
\makeatother

\DeclareMathOperator{\diam}{diam}
\DeclareMathOperator{\td}{td}
\DeclareMathOperator{\rc}{rc}
\DeclareMathOperator{\src}{src}
\DeclareMathOperator{\rvc}{rvc}
\DeclareMathOperator{\srvc}{srvc}

\newtheorem{theorem}{Theorem}
\newtheorem{lemma}[theorem]{Lemma}
\newtheorem{corollary}[theorem]{Corollary}

\newlist{pcases}{enumerate}{1}
\setlist[pcases]{
  label=\textbf{Case~\arabic*:}\protect\thiscase.~,
  ref=\arabic*,
  align=left,
  labelsep=0pt,
  leftmargin=0pt,
  labelwidth=0pt,
  parsep=0pt
}
\newcommand{\case}[1][]{%
  \if\relax\detokenize{#1}\relax
    \def\thiscase{}%
  \else
    \def\thiscase{~#1}%
  \fi
  \item
}

\title{Further Hardness Results on Rainbow and\\Strong Rainbow Connectivity}
\author{Juho Lauri \\ Department of Mathematics \\ Tampere University of Technology \\ Korkeakoulunkatu 1, 33720 Tampere, Finland \\ \texttt{juho.lauri@tut.fi}}
\date{\today}

\begin{document}
\maketitle

\begin{abstract}
A path in an edge-colored graph is \textit{rainbow} if no two edges of it are colored the same. The graph is said to be \textit{rainbow connected} if there is a rainbow path between every pair of vertices. If there is a rainbow shortest path between every pair of vertices, the graph is \textit{strong rainbow connected}. We consider the complexity of the problem of deciding if a given edge-colored graph is rainbow or strong rainbow connected. These problems are called \textsc{Rainbow connectivity} and \textsc{Strong rainbow connectivity}, respectively. We prove both problems remain $\NP$\hyp{}complete on interval outerplanar graphs and $k$-regular graphs for $k \geq 3$. Previously, no graph class was known where the complexity of the two problems would differ. We show that for block graphs, which form a subclass of chordal graphs, \textsc{Rainbow connectivity} is $\NP$\hyp{}complete while \textsc{Strong rainbow connectivity} is in $\P$. We conclude by considering some tractable special cases, and show for instance that both problems are in $\XP$ when parameterized by tree-depth.

\smallskip
\noindent \textbf{Keywords:} rainbow connectivity, computational complexity
\end{abstract}

\section{Introduction}
Let $G$ be an edge-colored undirected graph that is simple and finite. A path in $G$ is \textit{rainbow} if no two edges of it are colored the same. The graph $G$ is \textit{rainbow connected} if there is a rainbow path between every pair of vertices. If there is a rainbow shortest path between every pair of vertices, $G$ is \textit{strong rainbow connected}. Clearly, a strong rainbow connected graph is also rainbow connected. The minimum number of colors needed to make $G$ rainbow connected is known as the \textit{rainbow connection number} and is denoted by $\rc(G)$. Likewise, the minimum number of colors needed to make $G$ strong rainbow connected is known as the \textit{strong rainbow connection number} and is denoted by $\src(G)$. The concept of rainbow connectivity was introduced by~\citet{Chartrand2008} in 2008, and it has applications in data transfer and networking. The \textit{diameter} of a graph, denoted by $\diam(G)$, is the largest distance between two vertices of $G$. Clearly, $\diam(G)$ is a lower bound for $\rc(G)$. On the other hand, a trivial upper bound for $\rc(G)$ is $m$, where $m$ is the number of edges in $G$. Finally, because each strong rainbow connected graph is also rainbow connected, we have that $\diam(G) \leq \rc(G) \leq \src(G) \leq m$. For less trivial bounds and more, we refer the reader to the books~\citep{Chartrand2008b, Li2012b}, or the recent survey~\citep{Li2012}.

A similar concept was introduced for vertex-colored graphs by~\citet{Krivelevich2010}. A vertex-colored graph $H$ is \textit{rainbow vertex-connected} if every pair of vertices is connected by a path whose internal vertices have distinct colors. The minimum number of colors needed to make $H$ rainbow vertex-connected is known as the \textit{rainbow vertex-connection number} and is denoted by $\rvc(H)$. \citet{Li2014} investigated the \textit{strong rainbow vertex-connection number} as a natural variant. A vertex-colored graph is \textit{strong rainbow vertex-connected} if every pair of vertices is connected by a shortest path whose internal vertices have distinct colors. The minimum number of colors needed to make $H$ strong rainbow vertex-connected is known as the \textit{strong rainbow vertex-connection number} and is denoted by $\srvc(H)$. For rainbow vertex-connection numbers or other rainbow connection numbers outside of our scope we refer the reader to~\citep{Li2012}.

Rainbow connectivity can be motivated by the following example from the domain of networking. Suppose we have a network of agents represented as a graph. Each vertex in the graph represents an agent, and an edge between two agents is a link. An agent in the network wishes to communicate with every other agent in the network by sending messages. A message sent from agent $A$ to agent $B$ is routed through other agents that act as intermediaries. This communication path uses links between agents, and each link uses a channel. For the message to get through, we require that each link on the communication path receives a distinct channel. Given a network of agents $G$, our objective is to ensure each pair of agents can establish a communication path, while also minimizing the number of channels needed. The minimum number of channels we need is exactly $\rc(G)$.

\citet{Chakraborty2009} showed that it is $\NP$\hyp{}complete to decide if $\rc(G) \leq k$ for $k=2$.~\citet{Ananth2011} proved the problem remains hard for $k \geq 3$ as well. \citet{Chandran2013} proved there is no polynomial time algorithm to rainbow color graphs with less than twice the optimum number of colors, unless $\P = \NP$. Computing the strong rainbow connection number is known to be hard as well. \citet{Chartrand2008} proved $\rc(G) = 2$ if and only if $\src(G) = 2$, so deciding if $\src(G) \leq k$ is $\NP$\hyp{}complete for $k=2$. \citet{Ananth2011} showed the problem remains $\NP$\hyp{}complete for $k \geq 3$ even when $G$ is bipartite~\citep{Ananth2011}. In the same paper, they also showed there is no polynomial time algorithm for approximating the strong rainbow connection number of an $n$-vertex graph within a factor of $n^{1/2-\epsilon}$, where $\epsilon > 0$ unless $\NP = ZPP$.

Given that it is hard to compute both the rainbow and the strong rainbow connection number, it is natural to ask if it is easier to verify if a given edge-colored graph is rainbow or strong rainbow connected. In this paper, we are concerned with the complexity of the following two decision problems:
\begin{framed}
\noindent \textsc{Rainbow connectivity} \\
\textbf{Instance:} An undirected graph $G=(V,E)$, and an edge-coloring $\chi : E \to C$, where $C$ is a set of colors \\ 
\textbf{Question:} Is $G$ rainbow connected under $\chi$?
\end{framed}
\begin{framed}
\noindent \textsc{Strong rainbow connectivity} \\
\textbf{Instance:} An undirected graph $G=(V,E)$, and an edge-coloring $\chi : E \to C$, where $C$ is a set of colors \\ 
\textbf{Question:} Is $G$ strong rainbow connected under $\chi$?
\end{framed}

\noindent Out of these two problems, \textsc{Rainbow connectivity} has gained considerably more attention in the literature. \citet{Chakraborty2009} observed the problem is easy when the number of colors $|C|$ is bounded from above by a constant. However, they proved that for an arbitrary coloring, the problem is $\NP$\hyp{}complete. Building on their result,~\citet{Li2011} proved \textsc{Rainbow connectivity} remains $\NP$\hyp{}complete for bipartite graphs. Furthermore, the problem is $\NP$\hyp{}complete even for bipartite planar graphs as shown by~\citet{Huang2011}. Recently,~\citet{Uchizawa2013} complemented these results by showing \textsc{Rainbow connectivity} is $\NP$\hyp{}complete for outerplanar graphs, and even for series-parallel graphs. In the same paper, the authors also gave some positive results. Namely, they showed the problem is in $\P$ for cactus graphs, which form a subclass of outerplanar graphs. Furthermore, they settled the precise complexity of the problem from a viewpoint of graph diameter by showing the problem is in $\P$ for graphs of diameter~1, but $\NP$\hyp{}complete already for graphs of diameter greater than or equal to~2. To the best of our knowledge,~\citet{Uchizawa2013} were the only ones to consider \textsc{Strong rainbow connectivity}. They showed the problem is in $\P$ for cactus graphs, but $\NP$\hyp{}complete for outerplanar graphs. We shortly mention similar hardness results are known for deciding if a given vertex-colored is rainbow vertex-connected (see e.g.\ \citep{Chen2011,Li2011,Uchizawa2013}).

A \textit{fixed-parameter algorithm} (FPT) solves a problem with an input instance of size $n$ and a parameter $k$ in $f(k) \cdot n^{O(1)}$ time for some computable function $f$ depending solely on $k$. That is, for every fixed parameter value it yields a solution in polynomial time and the degree of the polynomial is independent from $k$.~\citet{Uchizawa2013} gave FPT algorithms for both problems on general graphs when parameterized by the number of colors $k = |C|$. These algorithms run in $O(k2^kmn)$ time and $O(k2^kn)$ space, where $n$ and $m$ are the number of vertices and edges in the input graph, respectively. These algorithms imply both \textsc{Rainbow connectivity} and \textsc{Strong rainbow connectivity} are solvable in polynomial time for any $n$-vertex graph if $|C| = O(\log n)$.

In this paper, we prove both \textsc{Rainbow connectivity} and \textsc{Strong rainbow connectivity} remain $\NP$\hyp{}complete for interval outerplanar graphs. We then consider the class of block graphs, which form a subclass of chordal graphs. Interestingly, for block graphs \textsc{Rainbow connectivity} is $\NP$\hyp{}complete, while \textsc{Strong rainbow connectivity} is in $\P$. To the best of our knowledge, this is the first graph class known for which the complexity of these two problems differ. Both problems are easy on 2-regular graphs. However, we show that both problems become $\NP$\hyp{}complete on cubic graphs, and further generalize this for $k$-regular graphs, where $k > 3$. This completely settles the complexity of both problems from the viewpoint of regularity.

\section{Preliminaries}
All graphs in this paper are simple, finite, and undirected. We begin by defining the graph classes we consider in this work. For graph theoretic concepts not defined here, we refer the reader to~\citep{Diestel2005}. For an integer $n$, we write $[n] = \{1,2,\ldots,n\}$.

A \textit{chord} is an edge joining two non-consecutive vertices in a cycle. A graph is \textit{chordal} if every cycle of length 4 or more has a chord. Equivalently, a graph is chordal if it contains no induced cycle of length 4 or more. A \textit{cut vertex} is a vertex whose removal will disconnect the graph. A \textit{biconnected graph} is a connected graph having no cut vertices. A \textit{block graph} is an undirected graph where every maximal biconnected component, known as a \textit{block}, is a clique. In a block graph $G$, different blocks intersect in at most one vertex, which is a cut vertex of~$G$. In other words, every edge of $G$ lies in a unique block, and $G$ is the union of its blocks. It is easy to see that a block graph is chordal. Another well-known subclass of chordal graphs is formed by \textit{interval graphs}. To define such graphs, we will first introduce the notion of \textit{clique trees}. A clique tree of a connected chordal graph $G$ is any tree $T$ whose vertices are the maximal cliques of $G$ such that for every two maximal cliques $C_i,C_j$, each clique on the path from $C_i$ to $C_j$ in $T$ contains $C_i \cap C_j$. Chordal graphs are precisely the class of graphs that admit a clique tree representation~\citep{Gavril1974}. As shown by~\citet{Gilmore1964}, a graph is an interval graph if and only if it admits a clique tree that is a path. A graph is \textit{planar} if it can be embedded in the plane without crossing edges. A graph is \textit{outerplanar} if it has a crossing-free embedding in the plane such that all vertices are on the same face. Finally, the \textit{degree} of a vertex is the number of edges incident to it. A graph is $k$-regular if the degree of each of its vertices is exactly $k$. Specifically, a 3-regular graph is known as a \textit{cubic graph}.

The $3$\hyp{}\textsc{Occurrence} $3$-\textsc{SAT} problem is a variant of the 3\hyp{}\textsc{SAT} problem where every variable occurs at most three times. The $\NP$\hyp{}completeness of \textsc{Rainbow connectivity} and \textsc{Strong rainbow connectivity} for outerplanar graphs were shown by a reduction from the $3$\hyp{}\textsc{Occurrence} $3$-\textsc{SAT} problem by~\citet{Uchizawa2013}. Often, it does not matter if one refers by 3\hyp{}\textsc{SAT} to the variant of 3\hyp{}\textsc{SAT} where each clause has exactly 3 literals, or the variant where each clause has at most 3 literals since both are $\NP$\hyp{}complete. However, for $3$\hyp{}\textsc{Occurrence} $3$-\textsc{SAT} this distinction is crucial. The variant where every clause has exactly 3 literals is in $\P$ because every such instance is satisfiable as shown by~\citet{Tovey1984}. The variant where every clause has at most 3 literals is however $\NP$\hyp{}complete~\citep{Papadimitriou1994}.

This distinction is not explicitly made by~\citet{Uchizawa2013}. However, it is not hard to modify their clause gadgets to allow for less than 3 literals. In other words, this does not affect the correctness of their reductions. The clause gadgets corresponding to clauses of size one and two can be found from the Appendix of this article. Their reductions greatly inspire ours, and thus we also reduce from the $3$\hyp{}\textsc{Occurrence} $3$-\textsc{SAT} problem, where each clause has at most 3 literals. We begin by describing their construction as our reductions are based on it. We tighten their result slightly by observing the resulting graph is both bipartite and outerplanar.

\begin{theorem}[{\citet{Uchizawa2013}}]
\label{thm_rc_is_npc_for_planar_graphs}
\textsc{Rainbow connectivity} is $\NP$\hyp{}complete when restricted to the class of bipartite outerplanar graphs.
\end{theorem}
\textbf{Construction}: We first observe the problem is in $\NP$ with the certificate being a set of colored paths, one for each pair of vertices. It is then simple to decide if a given path is rainbow. Given a $3$\hyp{}\textsc{Occurrence} $3$\hyp{}\textsc{SAT} formula $\phi = \bigwedge_{j=1}^{m} c_i$ over variables $x_1,x_2,\ldots,x_n$, we construct a graph $G_\phi$ and an edge-coloring $\chi$ such that $\phi$ is satisfiable if and only if $G_\phi$ is rainbow connected under $\chi$. We first describe the construction of $G_\phi$, and then the edge-coloring $\chi$ of $G_\phi$.

\begin{figure}
\bsubfloat[]{%
  \includegraphics[scale=1]{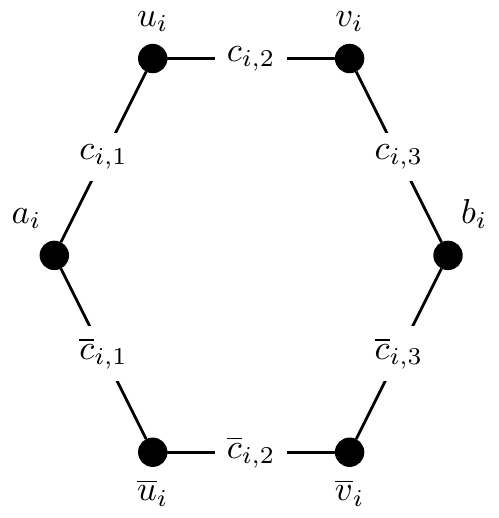}%
}
\bsubfloat[]{%
  \includegraphics[scale=1]{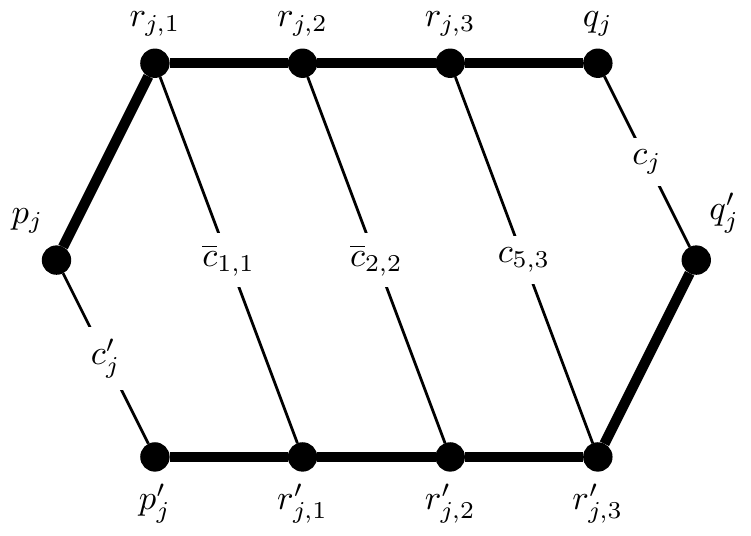}%
}
\bsubfloati\qquad\bsubfloatii
\caption{\textbf{(a)} A variable gadget $X_i$ for the variable $x_i$, and \textbf{(b)} a clause gadget $C_j$ for the clause $c_j = (x_1 \vee x_2 \vee \neg x_5)$, where $x_1$ is the first literal of $x_1$, $x_2$ is the second literal of $x_2$, and $\neg x_5$ is the third literal of $x_5$.}
\label{fig_chordal_var_clause_gadget}
\end{figure}

For each variable $x_i$, $i \in [n]$, we build a \textit{variable gadget} $X_i$. A variable gadget $X_i$ is a cycle graph $C_6$ embedded in the plane on vertices $a_i$, $u_i$, $v_i$, $b_i$, $\overline{v}_i$, $\overline{u}_i$ in clockwise order. For each clause $c_j$, $j \in [m]$, we build a \textit{clause gadget} $C_j$. A clause gadget $C_j$ is built by starting from a cycle graph $C_{10}$ embedded in the plane on vertices $p_j$, $r_{j,1}$, $r_{j,2}$, $r_{j,3}$, $q_j$, $q'_j$, $r'_{j,3}$, $r'_{j,2}$, $r'_{j,1}$, $p'_j$ in clockwise order, and by adding chords $(r_{j,1},r'_{j,1})$, $(r_{j,2},r'_{j,2})$, and $(r_{j,3},r'_{j,3})$. These chords correspond to the three literals the clause $c_j$ has. Both a variable gadget and a clause gadget are shown in Figure~\ref{fig_chordal_var_clause_gadget}.

We connect $X_i$ with $X_{i+1}$ by adding an edge $(b_i,a_{i+1})$ for each $1 \leq i < n$. Then, we connect $C_j$ with $C_{j+1}$ by adding an edge $(q'_j,p_{j+1})$ for each $1 \leq j < m$. Likewise, we connect the two components together by adding the edge $(b_n,p_1)$. We then add one vertex $t$, and the edge $(q'_m,t)$. Finally, we build a path of length $m$ on vertices $s_1,s_2,\ldots,s_m$, and connect it to $G_\phi$ by adding the edge $(s_m,a_1)$. This completes the construction of $G_\phi$. We can verify $G_\phi$ is indeed a bipartite outerplanar graph.

We now describe the edge-coloring $\chi$ given to the edges of $G_\phi$. Notice there are exactly two paths between $a_i$ and $b_i$ in a variable gadget $X_i$. Intuitively, taking the path from $a_i$ to $b_i$ through $u_i$ and $v_i$ corresponds to setting $x_i = 1$ in the formula $\phi$. We refer to this path as the \textit{positive $X_i$ path}. We color the three edges $(a_i,u_i), (u_i,v_i)$, and $(v_i,b_i)$ with colors $c_{i,1}, c_{i,2}$, and $c_{i,3}$, respectively. Taking the path from $a_i$ to $b_i$ through $\overline{u}_i$ and $\overline{v}_i$ corresponds to setting $x_i = 0$ in the formula $\phi$. We refer to this path as the \textit{negative $X_i$ path}. The three edges $(a_i,\overline{u}_i),(\overline{u}_i,\overline{v}_i)$ and $(\overline{v}_i,b_i)$ receive the colors $\overline{c}_{i,1},\overline{c}_{i,2}$ and $\overline{c}_{i,3}$, respectively. The coloring of a variable gadget $X_i$ is illustrated in Figure~\ref{fig_chordal_var_clause_gadget} (a).

Recall a variable $x_i$ appears at most three times in $\phi$. We refer to the first occurrence of $x_i$ as the \textit{first literal of $x_i$}, the second occurrence of $x_i$ as the \textit{second literal of $x_i$}, and finally the third occurrence of $x_i$ as the \textit{third literal of $x_i$}. If a clause has two or three literals of a same variable, the tie is broken arbitrarily. In a clause gadget $C_j$, we color the edge $(p_j,p'_j)$ with the color $c_j'$, and the edge $(q_j,q'_j)$ with the color $c_j$. For each $k \in [3]$, we denote the $k$th literal in the $j$th clause by $l_{j,k}$. We color the edge $(r_{j,k},r'_{j,k})$ as follows:
\[
 \chi((r_{j,k},r'_{j,k})) = 
 \begin{dcases*}
        \overline{c}_{i,1} 	& if $l_{j,k}$ is a positive literal and the first literal of $x_i$ \\
        \overline{c}_{i,2} 	& if $l_{j,k}$ is a positive literal and the second literal of $x_i$ \\
        \overline{c}_{i,3} 	& if $l_{j,k}$ is a positive literal and the third literal of $x_i$ \\
        c_{i,1} 		& if $l_{j,k}$ is a negative literal and the first literal of $x_i$ \\
        c_{i,2} 		& if $l_{j,k}$ is a negative literal and the second literal of $x_i$ \\
        c_{i,3} 		& if $l_{j,k}$ is a negative literal and the third literal of $x_i$
  \end{dcases*}
\]
The edge $(q'_j,p_{j+1})$, for each $1 \leq j < m$, receives the color $c_j'$, while the edge $(q'_m,t)$ is colored with $c'_m$. The coloring of a clause gadget $C_j$ is shown in Figure~\ref{fig_chordal_var_clause_gadget} (b).

Finally, we color each edge $(s_j,s_{j+1})$ with the color $c_j$ for each $1 \leq j < m$. The edge $(s_m,a_1)$ is colored with the color $c_m$. Every other uncolored edge of $G_\phi$ receives a fresh new color, that does not appear in $G_\phi$. Formally, these are precisely the edges in $U \cup W$, where $U = \{ (b_i, a_{i+1}) \mid 1 \leq i < n \} \cup \{ (b_n, p_1) \}$ and
\begin{equation*}
\begin{split}
W &= \{ (p_j,r_{j,1}), (r_{j,1},r_{j,2}), (r_{j,2},r_{j,3}), (r_{j,3},q_j) \mid 1 \leq j \leq m \} \\
\quad &\cup \{(q'_j,r'_{j,3}), (r'_{j,3},r'_{j,2}), (r'_{j,2},r'_{j,1}), (r'_{j,1},p'_j) \mid 1 \leq j \leq m \}.
\end{split}
\end{equation*}
The edges in $W$ correspond precisely to the edges drawn with thick lines in Figure~\ref{fig_chordal_var_clause_gadget} (b), for each clause gadget $C_j$. This completes the edge-coloring $\chi$ of $G_\phi$. The following claim is true for $G_\phi$, and it furthermore proves Theorem~\ref{thm_rc_is_npc_for_planar_graphs}.
\begin{lemma}[{\citet{Uchizawa2013}}]
\label{lemma_g_rainbow_conn}
The graph $G_\phi$ is rainbow connected under $\chi$ if and only if $G_\phi$ has a rainbow path between the vertices $s_1$ and $t$. Furthermore, there is a rainbow path between $s_1$ and $t$ if and only if the formula $\phi$ is satisfiable.
\end{lemma}

In the previous reduction, by observing every pair of vertices is rainbow connected by a rainbow shortest path given an satisfiable instance of $\phi$,~\citet{Uchizawa2013} also got the following.
\begin{theorem}[{\citet{Uchizawa2013}}]
\textsc{Strong rainbow connectivity} is $\NP$\hyp{}complete when restricted to the class of bipartite outerplanar graphs.
\end{theorem}

\section{Hardness results}
In this section, we give new hardness results for both \textsc{Rainbow connectivity} and \textsc{Strong rainbow connectivity}. All of our hardness results will follow by a reduction from the $3$\hyp{}\textsc{Occurrence} $3$\hyp{}\textsc{SAT} problem, and will essentially be based on Theorem~\ref{thm_rc_is_npc_for_planar_graphs}. For the sake of brevity, and similarly to Theorem~\ref{thm_rc_is_npc_for_planar_graphs}, we will present our constructions assuming each clause is of size three. The clause gadgets corresponding to clauses of size one and two can be found in the Appendix for each graph class.

We summarize the known complexity results for both problems in Table~\ref{table_hardness_summary} along with our new results.
\begin{table}
\caption{Summary of known complexity results for \textsc{Rainbow connectivity} and \textsc{Strong rainbow connectivity}. The symbol $\dagger$ stands for~\citep{Uchizawa2013}, and the symbol $\star$ for~\citep{Huang2011}.}
\label{table_hardness_summary}
\centering
\begin{tabular}{lll}
\toprule 
Graph class & \textsc{Rainbow connectivity} & \textsc{Strong rainbow connectivity} \\ 
\midrule
Bounded diameter $\geq 2$ 	& $\NP$-complete $\dagger$ & $\P$ [Theorem~\ref{thm_src_bounded_diam}] \\ 
Series-parallel 			& $\NP$-complete $\dagger$ & $\NP$-complete $\dagger$ \\
Bipartite planar 			& $\NP$-complete $\star$ & $\NP$-complete $\dagger$ \\
Bipartite outerplanar 		& $\NP$-complete $\dagger$ & $\NP$-complete $\dagger$ \\
Interval outerplanar		& $\NP$-complete [Theorem~\ref{thm_rc_is_hard_for_chordal}] & $\NP$-complete [Corollary~\ref{thm_src_is_npc_for_chordal}] \\
Cactus 						& $\P$ $\dagger$ & $\P$ $\dagger$ \\
$k$-regular, $k \geq 3$		& $\NP$-complete [Theorem~\ref{thm_rc_npc_kregular}] & $\NP$-complete [Corollary~\ref{thm_src_npc_kregular}] \\
Block 						& $\NP$-complete & $\P$ [Corollary~\ref{corollary_src_easy_for_block_graphs}] \\
Interval block				& $\NP$-complete [Theorem~\ref{thm_rc_is_npc_for_block}] & $\P$ \\
Tree 						& $\P$ & $\P$ \\
\bottomrule 
\end{tabular}
\end{table}

\subsection{Rainbow and strong rainbow connectivity are $\NP$\hyp{}complete for interval outerplanar graphs}
In this subsection, we prove \textsc{Rainbow connectivity} and \textsc{Strong rainbow connectivity} remain $\NP$\hyp{}complete for interval outerplanar graphs.

\begin{theorem}
\label{thm_rc_is_hard_for_chordal}
\textsc{Rainbow connectivity} is $\NP$\hyp{}complete when restricted to the class of interval outerplanar graphs.
\end{theorem}
\begin{proof}
We assume the same terminology as in Theorem~\ref{thm_rc_is_npc_for_planar_graphs}. Given a $3$\hyp{}\textsc{Occurrence} $3$\hyp{}\textsc{SAT} instance $\phi$, we first build a graph $G_\phi$ along with its edge-coloring $\chi$ precisely as in Theorem~\ref{thm_rc_is_npc_for_planar_graphs}. For clarity, we then rename $G_\phi$ to $G^M_\phi$, and $\chi$ to $\chi_M$. 

A variable gadget $X^M_i$ is obtained from $X_i$ by adding three chords $(u_i,\overline{u}_i)$, $(u_i,\overline{v}_i)$, and $(v_i,\overline{v}_i)$, and coloring each with a new color $\overline{c}_i$. Next, a clause gadget $C^M_j$ is obtained from $C_j$ by adding four chords $(r_{j,1},p'_j)$, $(r_{j,2},r'_{j,1})$, $(r_{j,3},r'_{j,2})$, and $(q_j,r'_{j,3})$. Each of these four chords receive the color $c'_j$. Finally, we recolor each edge in $U = \{ (b_i, a_{i+1}) \mid 1 \leq i < n \} \cup \{ (b_n,p_1) \}$ with the color $\overline{c}_i$. We can now verify that $G^M_\phi$ is indeed a chordal outerplanar graph. Furthermore, it is easy to see $G^M_\phi$ admits a clique tree that is a path. Thus, $G^M_\phi$ is both interval and outerplanar.

We then show these modifications do not contradict Lemma~\ref{lemma_g_rainbow_conn}. First, observe the distance between $a_i$ and $b_i$ for each $1 \leq i \leq n$ remains unchanged. However, we introduce additional paths between $a_i$ and $b_i$. But because every edge in $U$ is a bridge and has the color $\overline{c}_i$, it still holds that any rainbow path from $s_1$ to $t$ must, in every $X^M_i$, take precisely either the positive $X^M_i$ path or the negative $X^M_i$ path. 

Similarly, we also establish additional paths between $p_j$ and $q'_j$. However, because each edge in $\{ (q'_j,p_{j+1}) \mid 1 \leq j < m \} \cup \{ (q'_m,t)\}$ is a bridge and has the color $c'_j$, none of the newly added chords can be on a rainbow path from $s_1$ to $t$. Finally, observe also the distance between $p_j$ and $q'_j$ for each $1 \leq j \leq m$ remains unchanged. This implies any rainbow path from $s_1$ to $t$ must still, in every $C^M_j$, use precisely one of the edges $(r_{j,1},r'_{j,1})$, $(r_{j,2},r'_{j,2})$, or $(r_{j,3},r'_{j,3})$. Thus, Lemma~\ref{lemma_g_rainbow_conn} still holds, and we have the theorem.
\end{proof}
Similarly to Theorem~\ref{thm_rc_is_npc_for_planar_graphs}, given a satisfiable instance of $\phi$, we can observe there is a rainbow shortest path between every pair of vertices. Thus we get the following.
\begin{corollary}
\label{thm_src_is_npc_for_chordal}
\textsc{Strong rainbow connectivity} is $\NP$\hyp{}complete when restricted to the class of interval outerplanar graphs.
\end{corollary}

\subsection{Rainbow connectivity is $\NP$\hyp{}complete for interval block graphs}
In this subsection, we prove \textsc{Rainbow connectivity} is $\NP$\hyp{}complete for interval block graphs, which form a subclass of chordal graphs, and also generalize trees. It is worth noting that unlike in Theorems~\ref{thm_rc_is_npc_for_planar_graphs} and~\ref{thm_rc_is_hard_for_chordal}, the reduction we give next does not show hardness of \textsc{Strong rainbow connectivity} for block graphs.

\begin{theorem}
\label{thm_rc_is_npc_for_block}
\textsc{Rainbow connectivity} is $\NP$\hyp{}complete when restricted to the class of interval block graphs.
\end{theorem}
\begin{proof}
We assume the same terminology as in Theorem~\ref{thm_rc_is_hard_for_chordal}. Given a $3$\hyp{}\textsc{Occurrence} $3$\hyp{}\textsc{SAT} instance $\phi$, we first build a graph $G^M_\phi$ along with its edge-coloring $\chi_M$ precisely as in Theorem~\ref{thm_rc_is_hard_for_chordal}. For clarity, we rename $G^M_\phi$ to $G^B_\phi$, and $\chi_M$ to $\chi_B$.

We obtain an $X^B_i$ by adding to $X^M_i$ all the possible chords, that is, the edges $(a_i,\overline{v}_i)$, $(a_i,b_i)$, $(a_i,v_i)$, $(u_i,b_i)$, $(v_i,\overline{u}_i)$, and $(b_i,\overline{u}_i)$. Each of these chords receive the color $\overline{c}_i$. We also add all possible chords to every $C^M_j$, and thus obtain the clause gadget $C^B_j$. Formally, we add to $G^B_\phi$ the edges in
\begin{equation*}
\begin{split}
Z &= \{ (p_j,r_{j,2}), (p_j,r_{j,3}), (p_j,q_j), (p_j,q'_j), (p_j,r'_{j,3}), (p_j,r'_{j,2}), (p_j,r'_{j,1}) \mid 1 \leq j \leq m \} \\
\quad &\cup \{ (r_{j,1},r_{j,3}), (r_{j,1},q_j), (r_{j,1},q'_j), (r_{j,1},r'_{j,3}), (r_{j,1},r'_{j,2}) \mid 1 \leq j \leq m \} \\
\quad &\cup \{ (r_{j,2},q_j), (r_{j,2},q'_j), (r_{j,2},r'_{j,3}), (r_{j,2},p'_j) \mid 1 \leq j \leq m \} \\
\quad &\cup \{ (r_{j,3},q'_j), (r_{j,3},r'_{j,1}), (r_{j,3},p'_j) \mid 1 \leq j \leq m \} \\
\quad &\cup \{ (q_j,r'_{j,2}), (q_j,r'_{j,1}), (q_j,p'_j) \mid 1 \leq j \leq m \} \\
\quad &\cup \{ (q'_j,r'_{j,2}), (q'_j,r'_{j,1}), (q'_j, p'_j) \mid 1 \leq j \leq m \} \\
\quad &\cup \{ (r'_{j,3},r'_{j,1}), (r'_{j,3},p'_j) \mid 1 \leq j \leq m \} \\
\quad &\cup \{ (r'_{j,2},p'_j) \mid 1 \leq j \leq m \}.
\end{split}
\end{equation*}
Each edge in $Z$ receives the color $c'_j$. This completes the construction of $G^B_\phi$. Clearly, $G^B_\phi$ is now a block graph, with each block being a $K_2$, a $K_6$, or a $K_{10}$. Furthermore, it is easy to see $G^B_\phi$ admits a clique tree that is path. Thus, $G^B_\phi$ is both interval and block.

By an argument similar to Theorem~\ref{thm_rc_is_hard_for_chordal}, none of the newly added chords can be on a rainbow path from $s_1$ to $t$. Thus, Lemma~\ref{lemma_g_rainbow_conn} still holds, and we have the theorem.
\end{proof}

In the previous construction, the key difference to Theorem~\ref{thm_rc_is_hard_for_chordal} is that the distance between any pair of vertices in $X^B_i$ is one, as is the distance between any pair of vertices in $C^B_j$. Therefore, given a positive instance of $\phi$, it is not true that every pair of vertices in $G^B_\phi$ would be connected by a rainbow shortest path.

\subsection{Rainbow and strong rainbow connectivity are $\NP$\hyp{}complete for $k$-regular graphs}
In this subsection, we prove both \textsc{Rainbow connectivity} and \textsc{Strong rainbow connectivity} remain $\NP$\hyp{}complete for $k$-regular graphs, for $k \geq 3$. We begin by proving hardness for cubic graphs, that is, for $k = 3$. We use this construction as a building block for proving hardness for $k$-regular graphs, where $k > 3$.

\begin{figure}
\bsubfloat[]{%
  \includegraphics[scale=1]{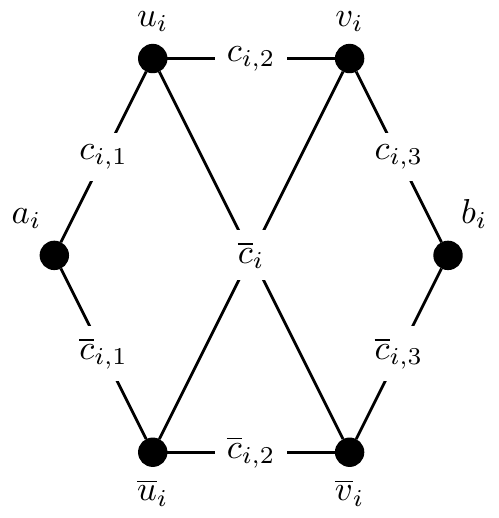}%
}
\bsubfloat[]{%
  \includegraphics[scale=1]{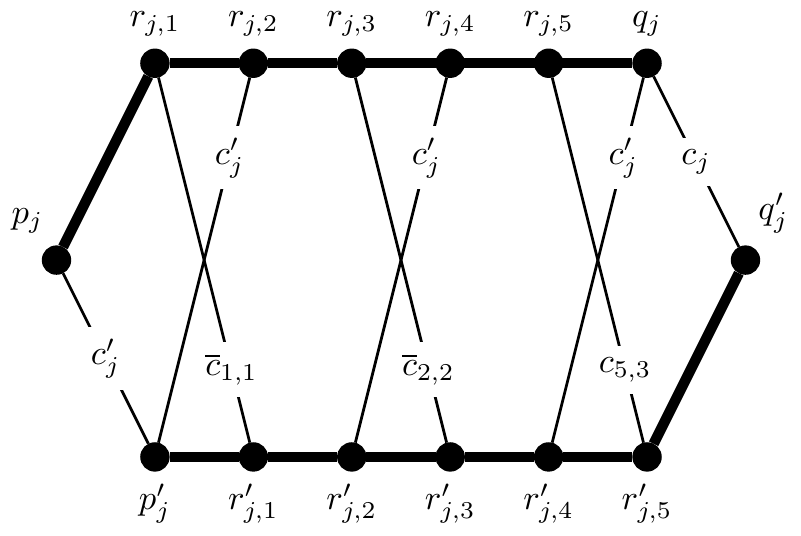}%
}
\bsubfloati\qquad\bsubfloatii
\caption{\textbf{(a)} A variable gadget $X^\Delta_i$ for the variable $x_i$, and \textbf{(b)} a clause gadget $C^\Delta_j$ for the clause $c_j = (x_1 \vee x_2 \vee \neg x_5)$, where $x_1$ is the first literal of $x_1$, $x_2$ is the second literal of $x_2$, and $\neg x_5$ is the third literal of $x_5$.}
\label{fig_cubic_var_clause_gadget}
\end{figure}

\begin{theorem}
\label{thm_rc_npc_cubic}
\textsc{Rainbow connectivity} is $\NP$\hyp{}complete when restricted to the class of cubic graphs.
\end{theorem}
\begin{proof}
We assume the terminology of Theorem~\ref{thm_rc_is_npc_for_planar_graphs}. Given a $3$\hyp{}\textsc{Occurrence} $3$\hyp{}\textsc{SAT} instance $\phi$, we first construct a graph $G^\Delta_\phi$, and then its edge-coloring $\chi_\Delta$.

We begin very similarly to Theorem~\ref{thm_rc_is_npc_for_planar_graphs}. A variable gadget $X^\Delta_i$ is built for every variable $x_i$, $i \in [n]$, by starting from an $X_i$ and adding two chords $(u_i,\overline{v_i})$ and $(\overline{u}_i,v_i)$. For each clause $c_j$, $j \in [m]$, we build a clause gadget $C^\Delta_j$. A clause gadget $C^\Delta_j$ is built by starting from a cycle graph $C_{14}$ embedded in the plane on vertices $p_j$, $r_{j,1}$, $r_{j,2}$, $r_{j,3}$, $r_{j,4}$, $r_{j,5}$, $q_j$, $q'_j$, $r'_{j,5}$, $r'_{j,4}$, $r'_{j,3}$, $r'_{j,2}$, $r'_{j,1}$, $p'_j$ in clockwise order, and by adding chords $(r_{j,1},r'_{j,1})$, $(r_{j,3},r'_{j,3})$, $(r_{j,5},r'_{j,5})$, $(p'_j,r_{j,2})$, $(r'_{j,2},r_{j,4})$, and $(r'_{j,4},q_j)$. The chords $(r_{j,1},r'_{j,1})$, $(r_{j,3},r'_{j,3})$, and $(r_{j,5},r'_{j,5})$ correspond to the three literals each clause has. Both a variable gadget and a clause gadget are shown in Figure~\ref{fig_cubic_var_clause_gadget}.

We then construct a \textit{tail gadget}, which is done by starting with two path graphs on $m-1$ vertices $s_1,\ldots,s_{m-1}$ and $s'_1,\ldots,s'_{m-1}$, respectively. Then, we add the edges $(s_j,s'_j)$ for each $3 \leq j \leq m-1$, and three edges $(s_1,s'_1)$, $(s'_1,s_2)$, and $(s_1,s'_2)$. Finally, we add a vertex $a_0$, and two edges $(s_{m-1},a_0)$ and $(s'_{m-1},a_0)$. The last gadget we build is a \textit{head gadget}. A head gadget is built by starting from a $K_4$ on vertices $t_1$, $t_2$, $t_3$, and $t_4$ with the edge $(t_1,t_2)$ removed. We then add the vertex $t_0$, and finally the edges $(t_0,t_1)$ and $(t_0,t_2)$. Both a tail gadget and a head gadget are shown in Figure~\ref{fig_tail_head_gadgets}.

We connect $X^\Delta_i$ with $X^\Delta_{i+1}$ by adding an edge $(b_i,a_{i+1})$ for each $1 \leq i < n$. Then, we connect $C^\Delta_j$ with $C^\Delta_{j+1}$ by adding an edge $(q'_j,p_{j+1})$ for each $1 \leq j < m$. These two components are connected by adding the edge $(b_n,p_1)$. The head gadget is connected to $G^\Delta_\phi$ by adding the edge $(t_0,q'_m)$, and the tail gadget by adding the edge $(a_0,a_1)$. This completes the construction of $G^\Delta_\phi$. We can now verify that $G^\Delta_\phi$ is indeed cubic.

We then describe the edge-coloring $\chi_\Delta$ of $G^\Delta_\phi$. The positive $X^\Delta_i$ path and the negative $X^\Delta_i$ path are colored precisely as in Theorem~\ref{thm_rc_is_npc_for_planar_graphs}. The two chords $(u_i,\overline{v}_i)$ and $(\overline{u}_i,v_i)$ receive the color $\overline{c}_i$, as does each edge in $U = \{ (b_i, a_{i+1}) \mid 1 \leq i < n \} \cup \{ (b_n,p_1) \}$. The coloring of a variable gadget $X^\Delta_i$ is illustrated in Figure~\ref{fig_cubic_var_clause_gadget} (a).

In a clause gadget $C^\Delta_j$, we color the edge $(p_j,p'_j)$ with the color $c_j'$, and the edge $(q_j,q'_j)$ with the color $c_j$. The three chords $(p'_j,r_{j,2})$, $(r'_{j,2},r_{j,4})$, and $(r'_{j,4},q_j)$ are colored with the color $c'_j$. For each $k \in \{1,2,3\}$, we color the edge $(r_{j,{2k-1}},r'_{j,{2k-1}})$ as follows:
\[
 \chi_\Delta((r_{j,k},r'_{j,k})) = 
 \begin{dcases*}
        \overline{c}_{i,1} 	& if $l_{j,k}$ is a positive literal and the first literal of $x_i$ \\
        \overline{c}_{i,2} 	& if $l_{j,k}$ is a positive literal and the second literal of $x_i$ \\
        \overline{c}_{i,3} 	& if $l_{j,k}$ is a positive literal and the third literal of $x_i$ \\
        c_{i,1} 		& if $l_{j,k}$ is a negative literal and the first literal of $x_i$ \\
        c_{i,2} 		& if $l_{j,k}$ is a negative literal and the second literal of $x_i$ \\
        c_{i,3} 		& if $l_{j,k}$ is a negative literal and the third literal of $x_i$
  \end{dcases*}
\]
The edge $(q'_j,p_{j+1})$, for each $1 \leq j < m$, receives the color $c'_j$, while the edge $(q'_m,t_0)$ is colored with $c'_m$. The coloring of a clause gadget $C^\Delta_j$ is shown in Figure~\ref{fig_cubic_var_clause_gadget} (b).

For each $1 \leq j < m-1$, we color the edge $(s_j,s_{j+1})$ with the color $c_j$, and also the edge $(s'_j,s'_{j+1})$ with the color $c_j$. The edges $(s_1,s'_2)$ and $(s_2,s'_1)$ both receive the color $c_1$. The edges $(s_{m-1},a_0)$ and $(s'_{m-1},a_0)$ both receive the color $c_{m-1}$. The bridge $(a_0,a_1)$ receives the color $c_m$. The coloring of a tail gadget is shown in Figure~\ref{fig_tail_head_gadgets} (a). Every other uncolored edge of $G^\Delta_\phi$ receives a fresh new color, that does not appear in $G^\Delta_\phi$. Formally, these are precisely the edges in
\begin{equation*}
\begin{split}
Q &= \{ (p_j,r_{j,1}), (r_{j,1},r_{j,2}), (r_{j,2},r_{j,3}), (r_{j,3},r_{j,4}), (r_{j,4},r_{j,5}), (r_{j,5},q_j) \mid 1 \leq j \leq m \} \\
&\quad \cup \{(q'_j,r'_{j,5}), (r'_{j,5},r'_{j,4}), (r'_{j,4},r'_{j,3}), (r'_{j,3},r'_{j,2}), (r'_{j,2},r'_{j,1}), (r'_{j,1},p'_j) \mid 1 \leq j \leq m \} \\
&\quad \cup \{ (s_j, s'_j) \mid 3 \leq j \leq m-1 \} \\
&\quad \cup \{ (s_1,s'_1) \} \\
&\quad \cup \{ (t_0,t_1), (t_0,t_2), (t_1,t_3), (t_1,t_4), (t_2,t_3), (t_2,t_4), (t_3,t_4) \}.
\end{split}
\end{equation*}

The edges in $Q$ correspond precisely to the edges drawn with thick lines in Figures~\ref{fig_cubic_var_clause_gadget} and~\ref{fig_tail_head_gadgets}. This completes the edge-coloring $\chi_\Delta$ of $G^\Delta_\phi$.

Let us rename $t_0$ as $t$. By an argument similar to Theorem~\ref{thm_rc_is_hard_for_chordal}, we can show there is a rainbow path between $s_1$ and $t$ (and similarly between $s'_1$ and $t$) if and only if $\phi$ is satisfiable.
\end{proof}

\begin{figure}
\bsubfloat[]{%
  \includegraphics[scale=1]{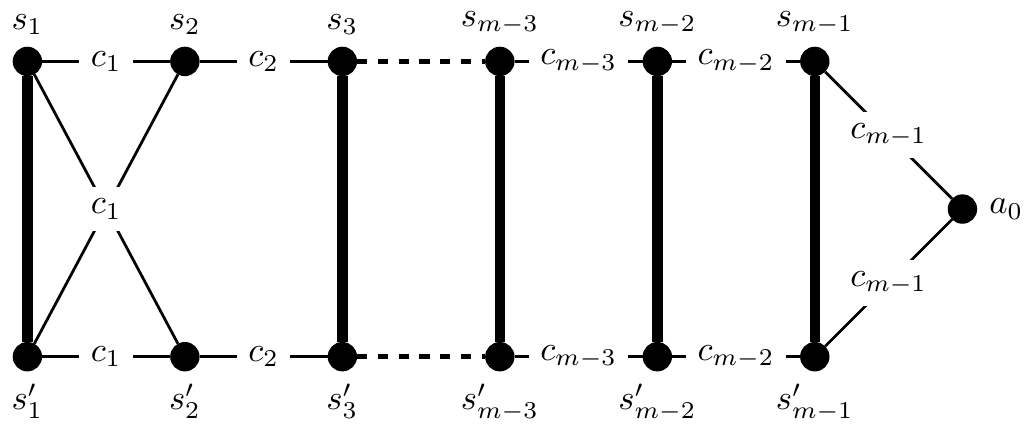}%
}
\bsubfloat[]{%
  \includegraphics[scale=1]{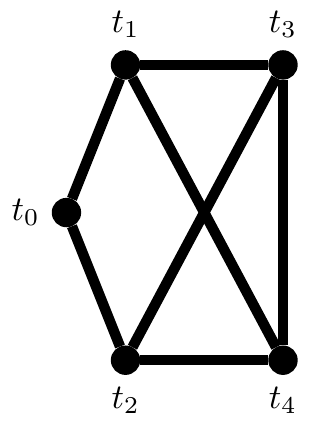}%
}
\bsubfloati\qquad\bsubfloatii
\caption{\textbf{(a)} A tail gadget, and \textbf{(b)} a head gadget.}
\label{fig_tail_head_gadgets}
\end{figure}

\noindent Again, in the positive case, every pair of vertices has a rainbow shortest path between them further giving us the following.
\begin{corollary}
\label{thm_src_npc_cubic}
\textsc{Strong rainbow connectivity} is $\NP$\hyp{}complete when restricted to the class of cubic graphs.
\end{corollary}

We are now ready to prove the hardness of both problems for $k$-regular graphs, where $k > 3$.

\begin{theorem}
\label{thm_rc_npc_kregular}
\textsc{Rainbow connectivity} is $\NP$\hyp{}complete when restricted to the class of $k$-regular graphs, where $k > 3$.
\end{theorem}
\begin{proof}
We assume the terminology of Theorem~\ref{thm_rc_npc_cubic}. Given a $3$\hyp{}\textsc{Occurrence} $3$\hyp{}\textsc{SAT} instance $\phi$, we construct a graph $G^*_\phi$, and its edge-coloring $\chi_*$.

We first construct $k-2$ copies of the cubic graph $G^\Delta_\phi$. Let us denote these copies as $G^\Delta_{\phi,h}$, where $h \in [k-2]$. Each $G^\Delta_{\phi,h}$ retains its original coloring as defined in Theorem~\ref{thm_rc_npc_cubic}. That is, each $G^\Delta_{\phi,h}$ has precisely the same coloring. Let us assign a labeling $1,\ldots,|V(G^\Delta_\phi)|$ on the vertices of $G^\Delta_\phi$, and use the same labeling for each $G^\Delta_{\phi,h}$. By $v_{h,l}$ we denote the vertex in subgraph $G^\Delta_{\phi,h}$ with the label $l$, where $l \in [|V(G^\Delta_\phi)|]$. We then form a clique between the vertices $v_{h,l}$ for each $h$ and $l$ by adding all possible ${k-2 \choose 2}$ edges. These newly added edges are precisely the uncolored edges of $G^*_\phi$, and all of them receive the fresh new color $c^*$. Because $G^\Delta_\phi$ is cubic, we can verify $G^*_\phi$ is now $k$-regular. This completes the construction of both $G^*_\phi$, and its edge-coloring $\chi_*$.

We will then show $G^*_\phi$ is rainbow connected if and only if $\phi$ is satisfiable. Recalling the naming of vertices from Theorem~\ref{thm_rc_npc_cubic}, without loss let us rename $s_1$ in $G^\Delta_{\phi,1}$ as $s$, and $t_0$ in $G^\Delta_{\phi,1}$ as $t$. First suppose $\phi$ is satisfiable. Then because there is a rainbow path between $s$ and each vertex of $G^\Delta_{\phi,1}$ by Theorem~\ref{thm_rc_npc_cubic}, the graph $G^*_\phi$ is rainbow connected. Finally, suppose $\phi$ is unsatisfiable. Observe that any rainbow path from $s$ to $t$ must only consist of edges in $G^\Delta_{\phi,1}$. But since $s$ and $t$ are not rainbow connected, it follows that $G^*_\phi$ is not rainbow connected. Thus, we have the theorem.
\end{proof}
Again, the following is immediate from the previous construction.
\begin{corollary}
\label{thm_src_npc_kregular}
\textsc{Strong rainbow connectivity} is $\NP$\hyp{}complete when restricted to the class of $k$-regular graphs, where $k > 3$.
\end{corollary}

\section{Polynomial time solvable cases}
In this section, we consider \textsc{Strong rainbow connectivity} from a structural perspective. We observe some graph classes for which the problem is easy. We begin by showing bounding the diameter of the input graph makes \textsc{Strong rainbow connectivity} tractable, while this not so for \textsc{Rainbow connectivity}~\citep{Uchizawa2013}.

\begin{theorem}
\label{thm_src_bounded_diam}
\textsc{Strong rainbow connectivity} is solvable in $O(n^{d+3})$ time for graphs of bounded diameter~$d \geq 1$, where $n$ is the number of vertices in the input graph.
\end{theorem}
\begin{proof}
For $d=1$, the problem is trivial. So suppose $d \geq 2$, and let $n$ denote the number of vertices in $G$. Let $u$ and $v$ be two arbitrary vertices of $G$, and let $P = ut_1t_2 \cdots t_{d-1}v$ be a shortest path from $u$ to $v$. Because there are less than $n$ choices for each $t_i$ where $i \in [d-1]$, it follows that there are at most $n^d$ shortest $u$-$v$ paths of length no more than $d$. We can then check all of these paths of length exactly $d(u,v)$, and verify if at least one such path is rainbow. Clearly, it takes $O(d)$ time to check one path. Because we have ${n \choose 2}$ pairs of vertices to check and $d$ is fixed, it follows that \textsc{Strong rainbow connectivity} can be decided in $O(n^{d+3})$ time for graphs of bounded diameter.
\end{proof}

If a graph $G$ has exactly one shortest path between any pair of vertices, $G$ is said to be \textit{geodetic}. A graph is \textit{$k$-geodetic} if there are at most $k$ shortest paths between any pair of vertices. In fact, it is an easy observation that the brute-force algorithm that checks every shortest path between a pair of vertices runs in polynomial time for $k$-geodetic graphs.
\begin{theorem}
\label{thm_src_easy_kgeodetic}
\textsc{Strong rainbow connectivity} is solvable in polynomial time when restricted to the class of $k$-geodetic graphs, where $k = O(\poly(n,m))$, and $n$ and $m$ are the number of vertices and edges in the input graph, respectively.
\end{theorem}
As shown by~\citet{Stemple1968}, a connected graph $G$ is geodetic if and only if every block of $G$ is geodetic. By observing that a complete graph is geodetic, we get the following corollary.
\begin{corollary}
\label{corollary_src_easy_for_block_graphs}
\textsc{Strong rainbow connectivity} is solvable in polynomial time when restricted to the class of block graphs.
\end{corollary}

Finally, we make make some observations about the reductions built in this work, and describe consequences for parameterized complexity. It follows from the work of~\citet{Uchizawa2013} that both \textsc{Rainbow connectivity} and \textsc{Strong rainbow connectivity} remain $\NP$-complete when parameterized by \textit{treewidth}. Informally, treewidth is a measure of how close a graph is to being a tree. \textit{Pathwidth} of a graph measures the closeness to a path. Pathwidth of a graph $G$ can be defined to being one less than the maximum clique size in an interval supergraph of $G$. The interval outerplanar graph we construct in Theorem~\ref{thm_rc_is_hard_for_chordal} has maximum clique size 3. It follows both \textsc{Rainbow connectivity} and \textsc{Strong rainbow connectivity} are $\NP$-complete for graphs of pathwidth 2. But we can be slightly more general, and show hardness for graphs of pathwidth $p \geq 2$. To see this, observe we can connect a clique of size at least 3 to the graph constructed in Theorem~\ref{thm_rc_is_hard_for_chordal}, and color its edges with a fresh new color. This might break the property of being outerplanar, but the graph definitely remains interval. Thus, we observe the following.
\begin{corollary}
\label{thm_hardness_pw}
Both \textsc{Rainbow connectivity} and \textsc{Strong rainbow connectivity} are $\NP$-complete for graphs of pathwidth $p$, for every $p \geq 2$.
\end{corollary}
By the result of~\citet{Kaplan1996}, the \textit{bandwidth} of a graph $G$ is one less than the maximum clique size of any \textit{proper interval} supergraph of $G$, chosen to minimize its clique number. Proper interval graphs are exactly the \textit{claw-free} interval graphs~\citep{Roberts1969}, where a claw is the complete bipartite graph $K_{1,3}$. The interval outerplanar graph we construct in Theorem~\ref{thm_rc_is_hard_for_chordal} can be observed to be claw-free. Combining this observation with the argument above, we can again be slightly more general.
\begin{corollary}
Both \textsc{Rainbow connectivity} and \textsc{Strong rainbow connectivity} are $\NP$-complete for graphs of bandwidth $b$, for every $b \geq 2$.
\end{corollary}
Recall a problem is said to be in $\XP$ if it can be solved in $O(n^{f(k)})$ time, where $n$ is the input size, $k$ a parameter, and $f$ some computable function. Theorem~\ref{thm_src_bounded_diam} proves \textsc{Strong rainbow connectivity} is in $\XP$ when parameterized by the diameter of the graph. This implies the problem is in $\XP$ for several other parameters, such as domination number, independence number, minimum clique cover, distance to clique, distance to cograph, distance to co-cluster, vertex cover number, distance to cluster, and cluster editing (see e.g.\ \citep{Komusiewicz2012} for a relationship of some parameters). Corollary~\ref{thm_hardness_pw} extends the known hardness barrier from treewidth to pathwidth. Pathwidth is upper bounded by \textit{tree-depth}, which is informally a measure of how close a graph is to being a star (that is, the $K_{1,n}$). As shown by~\citet{Nesetril2008}, the length of a longest path for every undirected graph $G$ is upper bounded by $2^{\td(G)}-2$, where $\td(G)$ denotes the tree-depth of $G$. Combining this result with Theorem~\ref{thm_src_bounded_diam}, we obtain the following.
\begin{corollary}
Both \textsc{Rainbow connectivity} and \textsc{Strong rainbow connectivity} are in $\XP$ when parameterized by tree-depth.
\end{corollary}

\section*{Acknowledgements}
The author thanks Henri Hansen, Mikko Lauri, and Keijo Ruohonen for helpful comments, and acknowledges the idea of Henri Hansen that led to Theorem~\ref{thm_rc_npc_kregular}. The author also thanks the referees for their useful comments.

\bibliographystyle{model1-num-names}
\bibliography{bibliography}

\section*{Appendix}
In Section~2, we presented a reduction from the $3$\hyp{}\textsc{Occurrence} $3$\hyp{}\textsc{SAT} problem to \textsc{Rainbow connectivity} due to~\citet{Uchizawa2013}. In this appendix, we give the missing critical details of their proof, as discussed in the beginning of the section. Namely, we show how clause gadgets corresponding to clauses of size one and two can be built in Theorem~\ref{thm_rc_is_npc_for_planar_graphs}. For completeness, we describe similar gadgets for Theorems~\ref{thm_rc_is_hard_for_chordal}, \ref{thm_rc_is_npc_for_block}, and \ref{thm_rc_npc_cubic}.

The clause gadgets for four different graph classes are shown in Figure~\ref{fig_clause_gadgets}. The first column denotes the graph class. The second column shows a clause gadget corresponding to a clause containing one literal, while the third column does the same for a clause having two literals. For clarity, the edges denoted by thin lines having no labels on row three correspond to chords colored with the color $c'_j$ (refer to Theorem~\ref{thm_rc_is_npc_for_block} for details). See the respective theorems for an explanation of other colors appearing on the edges.

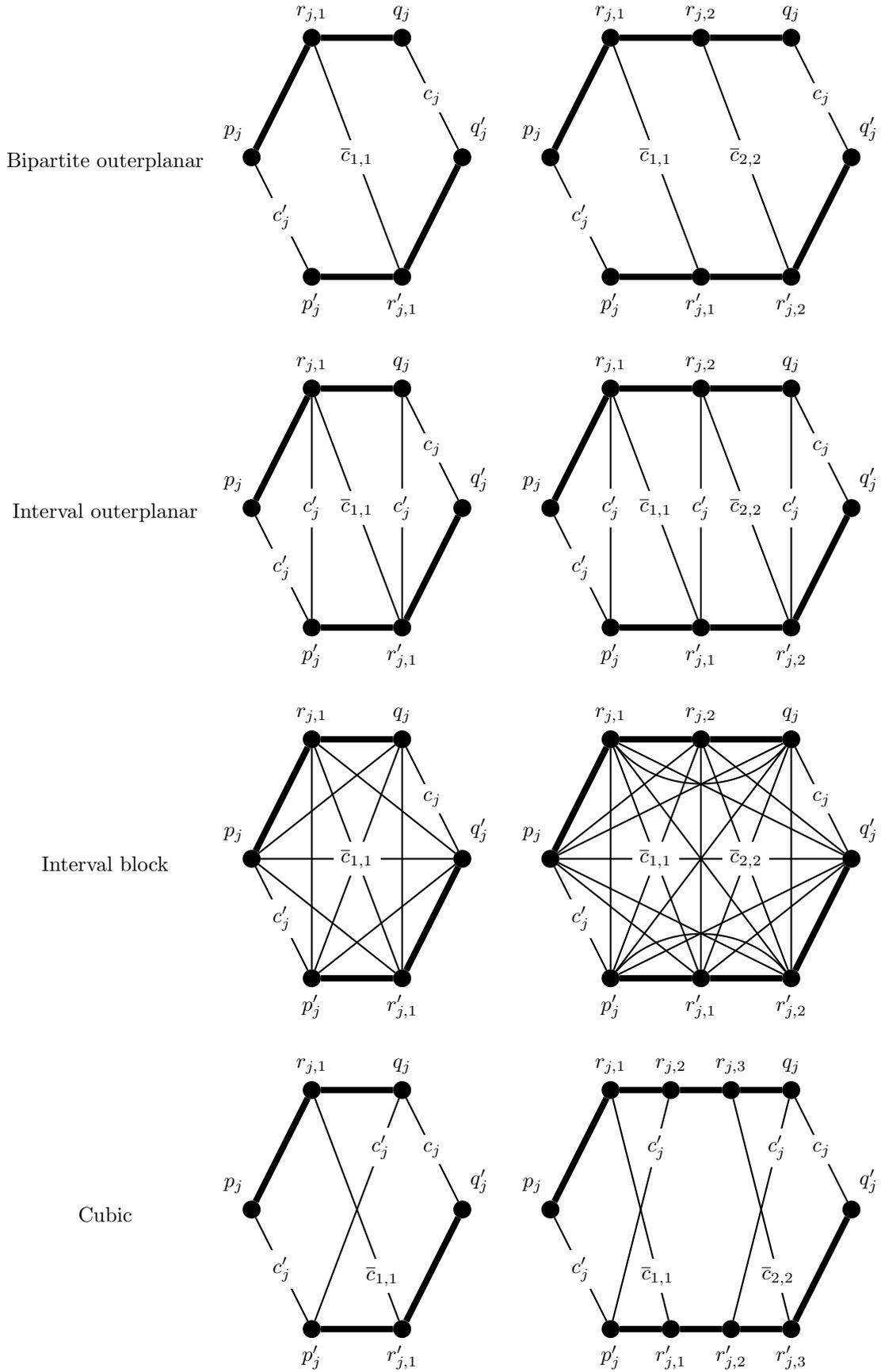
\begin{figure}
\begin{tikzpicture}[scale=2]
    \matrix[row sep=0.2cm,column sep=0.1cm]
    { 
    \node[draw=none,fill=none,inner sep=0pt] {Bipartite outerplanar}; &    
    \node(i1){\includegraphics[scale=1]{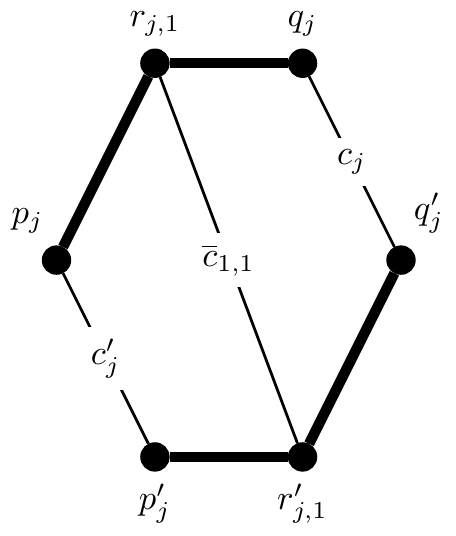}}; &
    \node(i2){\includegraphics[scale=1]{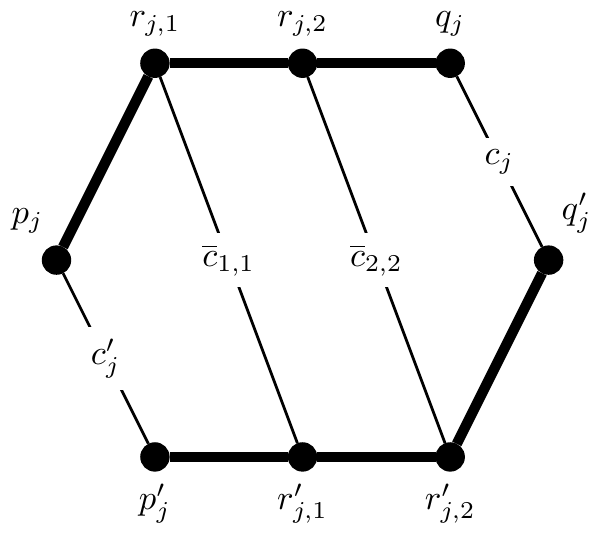}}; \\
 	
    \node[draw=none,fill=none,inner sep=0pt] {Interval outerplanar}; &
    \node(i3){\includegraphics[scale=1]{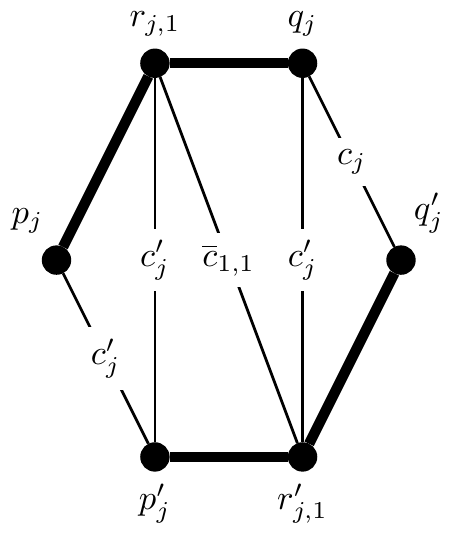}}; & 
 	\node(i4){\includegraphics[scale=1]{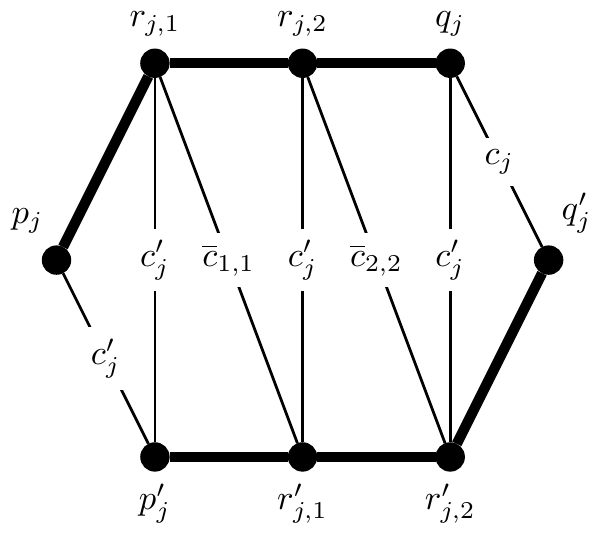}}; \\
 	
 	\node[draw=none,fill=none,inner sep=0pt] {Interval block}; &
    \node(i5){\includegraphics[scale=1]{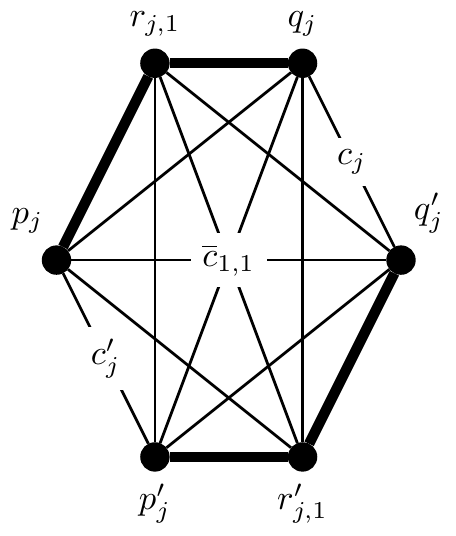}}; &
    \node(i6){\includegraphics[scale=1]{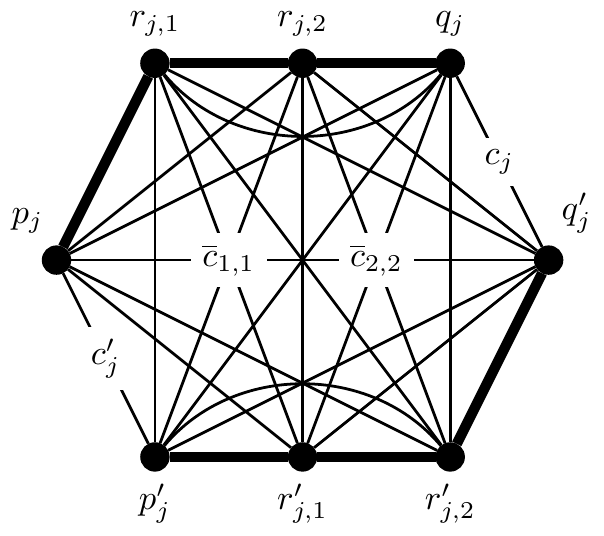}}; \\

	\node[draw=none,fill=none,inner sep=0pt] {Cubic}; &
	\node(i7){\includegraphics[scale=1]{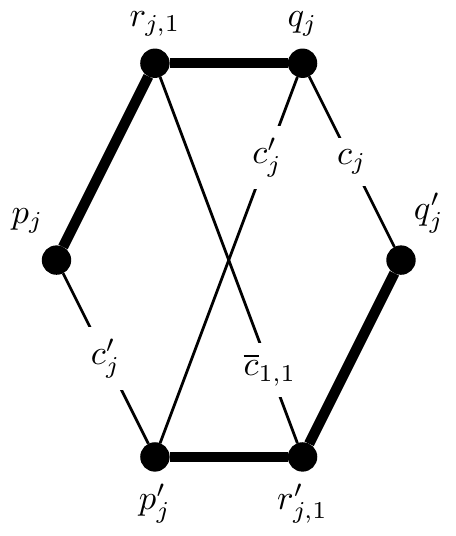}}; &
	\node(i8){\includegraphics[scale=1]{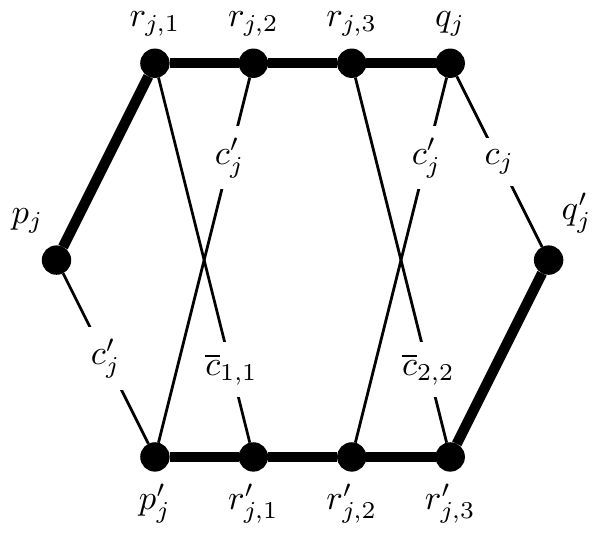}}; \\
 	};
\end{tikzpicture}
\caption{Clause gadgets corresponding to clauses of size one and two for different graph classes.}
\label{fig_clause_gadgets}
\end{figure}

\end{document}